\documentclass{llncs}

\usepackage[latin1]{inputenc}

\usepackage{indentfirst} 

\usepackage[english]{babel}
\usepackage{graphicx}
\usepackage{subfigure}
\usepackage{amssymb}

\usepackage{array}

\begin{document}

\title{A polynomial-time algorithm for planar multicuts with few source-sink pairs\thanks{This research work was supported by the French ANR project \emph{DOPAGE} (ANR-09-JCJC-0068)}}

\author{C\'edric Bentz}

\institute{LRI, Univ. Paris-Sud \& CNRS, 91405 Orsay Cedex, France\\\email{cedric.bentz@lri.fr}}



\spnewtheorem{Claim}{Claim}{\bfseries}{\itshape}

\maketitle

\thispagestyle{plain}

\begin{abstract}
Given an edge-weighted undirected graph and a list of $k$ source-sink pairs of vertices, the well-known \emph{minimum multicut problem} consists in selecting a minimum-weight set of edges whose removal leaves no path between every source and its corresponding sink. We give the first polynomial-time algorithm to solve this problem in planar graphs, when $k$ is fixed. Previously, this problem was known to remain \textbf{NP}-hard in general graphs with fixed $k$, and in trees with arbitrary $k$; the most noticeable tractable case known so far was in planar graphs with fixed $k$ and sources and sinks lying on the outer face.
\end{abstract}

\section{Introduction}\label{sectIntro}

In this paper, we are interested in the study of the minimum multicut problem in
undirected graphs (no directed version is considered). This fundamental problem has been
extensively studied, and is well-known to be $\textbf{NP}$-hard even in very restricted classes of graphs.

Assume we are given a $n$-vertex $m$-edge undirected graph $G=(V,E)$, a \emph{weight function} $w : E \rightarrow \mathbb{Z}^+$ and a list $\mathcal{L}$ of pairs (source $s_i$, sink $s'_i$) of \emph{terminal} vertices. Each pair $(s_i, s'_i)$ defines a \emph{commodity}. The \emph{minimum multicut problem} (\textsc{MinMC}) consists in selecting a minimum weight set of edges whose removal separates $s_i$ from $s'_i$ for each $i$. The \emph{minimum multiterminal cut problem} (\textsc{MinMTC}) is a special case of \textsc{MinMC} in which, given a set of vertices $\mathcal{T} = \{t_1, \dots, t_{\arrowvert \mathcal{T} \arrowvert}\}$, the source-sink pairs are $(t_i, t_j)$ for $i \neq j$.

For $\arrowvert \mathcal{L} \arrowvert = 1$, the problem is the classical minimum cut problem. For $\arrowvert \mathcal{L} \arrowvert = 2$, the problem can be solved in polynomial time by solving two minimum cut problems \cite{refYan83}. However, Dahlhaus et al. showed that, for any fixed $\arrowvert \mathcal{L} \arrowvert \geq 3$, \textsc{MinMTC} (and hence \textsc{MinMC}) becomes $\textbf{NP}$-hard (and even $\textbf{APX}$-hard) in general graphs \cite{refDahlhaus94}. When $\arrowvert \mathcal{L} \arrowvert$ is not fixed, \textsc{MinMC} is $\textbf{APX}$-hard even in unweighted stars \cite{refGVY97} and $\textbf{NP}$-hard even in unweighted binary trees \cite{refCalinescu03}, while \textsc{MinMTC} is $\textbf{NP}$-hard in planar graphs \cite{refDahlhaus94}. We also mention that, in bounded tree-width graphs, \textsc{MinMTC} (resp. \textsc{MinMC}) is polynomial-time solvable when $\arrowvert \mathcal{L} \arrowvert$ is arbitrary \cite{refGuo08} (resp. when $\arrowvert \mathcal{L} \arrowvert$ is fixed \cite{refBentz08}). There have been recent results concerning FPT algorithms for \textsc{MinM(T)C}: however, the parameter considered in these papers is the size of the solution, and hence we shall not mention them here.

In their seminal paper, Dahlhaus et al. also showed that \textsc{MinMTC} can be solved in polynomial time in planar graphs if $\arrowvert \mathcal{L} \arrowvert$ is fixed, but they left as open three important questions: first, does \textsc{MinMTC} admit a polynomial-time approximation scheme (PTAS)? Second, is \textsc{MinMTC} FPT in planar graphs, if $\arrowvert \mathcal{L} \arrowvert$ is viewed as the parameter \cite{refDowney99}? Third, is \textsc{MinMC} also polynomial-time solvable in planar graphs if $\arrowvert \mathcal{L} \arrowvert$ is fixed? The first open question was recently addressed by Bateni et al. \cite{refBateni11}. The second one was even more recently addressed by Marx \cite{refMarx12}, and we answer the third question in this paper (while the case where all the sources and sinks lie on the outer face was already solved in \cite{refBentz09}).

It should be noticed that Hartvigsen \cite{refHartvigsen} and Yeh \cite{refYeh01} later provided other algorithms to solve \textsc{MinMTC} in planar graphs when $\arrowvert \mathcal{L} \arrowvert$ is fixed (none of them being FPT with respect to $\arrowvert \mathcal{L} \arrowvert$). Moreover, it was observed in \cite{refBentz06} and \cite{refBentzENDM09} that unfortunately the proof of Yeh's algorithm is not correct, and later it was proved in \cite{refCheung10} that the algorithm itself is not correct. The main mistake in the proof of this algorithm was to assume that, when replacing the boundary of any connected component by a minimum cut between some well-chosen vertices, we still obtain a single connected component. More recently, Marx and Klein gave an even faster algorithm to solve \textsc{MinMTC} in planar graphs when $\arrowvert \mathcal{L} \arrowvert$ is fixed \cite{refKM2012}, but Marx also managed to prove that, assuming the \emph{Exponential Time Hypothesis}, this problem is \emph{not} FPT with respect to $\arrowvert \mathcal{L} \arrowvert$ \cite{refMarx12}. This latter result immediately implies that \textsc{MinMC} in planar graphs is not FPT with respect to $\arrowvert \mathcal{L} \arrowvert$.

In this paper, we give an algorithm based, on the one hand, on a revised and generalized Yeh-like approach, and, on the other hand, on shortest homotopic paths methods, and show that this algorithm can be used to solve \textsc{MinMC} in polynomial time when the graph is planar and $\arrowvert \mathcal{L} \arrowvert$ is fixed. (Obviously, this also provides an alternative polynomial-time algorithm to solve MinMTC in planar graphs when $\arrowvert \mathcal{L} \arrowvert$ is fixed.) It is worth noticing that our major tool is a
new characterization of optimal solutions for this problem. Moreover, although homotopic routing methods have already been used to solve planar disjoint paths problems (see \cite{Schrijver90} and \cite{Schrijver91} for instance), to the best of our knowledge they have never been used to solve (multi)cut problems so far. (Our algorithm is not FPT, but the recent result of Marx \cite{refMarx12} implies that unfortunately this is essentially the best one can hope for.)

The paper is organized as follows. In Section \ref{sectStart}, we describe the starting point of our algorithm. Then, in Section \ref{sectPreliminaries}, we give some preliminary definitions and results, that will be useful in Section \ref{sectAlgo}. Finally, in Section \ref{sectAlgo}, we describe our algorithm, and prove its correctness.\\

\section{The starting point}\label{sectStart}

The first step of our algorithm is a simple idea presented in \cite{refBentz09}. Given a \textsc{MinMC} instance $I=(G=(V,E), w, \mathcal{L})$ and any of its optimal multicuts $C$, one can define the clustering of the terminals associated with the connected components of $G'=(V,E\setminus C)$ (we also say that this particular clustering \emph{induces} these connected components). The $i$th cluster of this clustering, denoted by $\mathcal{T}_i$, contains all the terminals lying in the $i$th connected component of $G'$. Once this clustering has been defined (although, so far, we need to know $C$ in order to do it), finding an optimal solution to $I$ is equivalent to removing a minimum-weight set of edges $C$ whose removal separates all the terminals in $\mathcal{T}_i$ from all the terminals in $\mathcal{T}_j$ for each $i \neq j$.

In this paper, we will refer to this problem as the \emph{minimum multi-cluster cut problem} (\textsc{MinMCC}). This problem has been defined as the \emph{Colored Multiterminal Cut problem} in \cite{refDahlhaus94}, where it is shown to be $\textbf{NP}$-hard in planar graphs, even with only four clusters (and it is claimed that this is also true for three clusters). Note that, in general graphs, \textsc{MinMCC} and \textsc{MinMTC} are equivalent, since from a \textsc{MinMCC} instance we can obtain an equivalent \textsc{MinMTC} instance by adding one new terminal vertex for each cluster, and linking all the terminals in this cluster (which will no longer be terminals in the \textsc{MinMTC} instance) to this new vertex by sufficiently heavy edges. However, this reduction does \emph{not} necessarily preserve planarity. Given a \textsc{MinMC} instance, we can build an equivalent \textsc{MinMCC} instance by enumerating all the possible clusterings of the terminals (such a clustering can contain up to $2 \vert \mathcal{L} \vert$ clusters): when $\vert \mathcal{L} \vert$ is fixed, this can be done in constant time, and so this yields the following lemma.

\begin{lemma}
When $\vert \mathcal{L} \vert$ is fixed, \textsc{MinMC} can be polynomially reduced to \textsc{MinMCC}, and this reduction preserves planarity.
\end{lemma}

Since we enumerate all the possible clusterings in order to guess the right one, we can also assume without loss of generality that the one we chose has the property that no clustering associated with an optimal solution induces more connected components than this one does. In other words, in the (planar) \textsc{MinMCC} instance we obtain, every cluster induces exactly one connected component in any optimal solution. In the remainder of the paper, we design an efficient algorithm to solve \textsc{MinMCC} in planar graphs when the sum of the sizes of the clusters is fixed (otherwise, from \cite{refDahlhaus94}, the problem is $\textbf{NP}$-hard); from the above enumeration argument, we can assume that every cluster induces only one connected component (note that this problem generalizes planar \textsc{MinMTC} with a fixed number of terminals). To do this, we will make use of some notions and results related to planarity, planar curves and planar duality, which we introduce in the next section.

\section{Preliminary definitions and results}\label{sectPreliminaries}

Throughout the paper, each time we consider a \textsc{MinMCC} instance in a planar graph $G$, we assume without loss of generality that $G$ is simple, loopless, connected (otherwise, we can consider each connected component independently), and even $2$-vertex-connected (from \cite{refBentz09}), but also that some planar embedding of $G$ is given. Recall that to any planar graph $G$ (embedded in the plane) we can associate a dual (planar) graph $G^*$: each face (including the outer face) of the initial (or primal) graph $G$ is associated with one vertex in the dual graph $G^*$, and there is an edge between two vertices in the dual graph iff the associated faces are adjacent (i.e., share an edge) in the primal graph. (If an edge belongs to only one face, then it corresponds to a loop in the dual graph.) As a consequence, there is a one-to-one correspondence between primal faces (resp. vertices) and dual vertices (resp. faces).

\begin{figure}[h]
    \begin{center}
        \includegraphics[width=11cm]{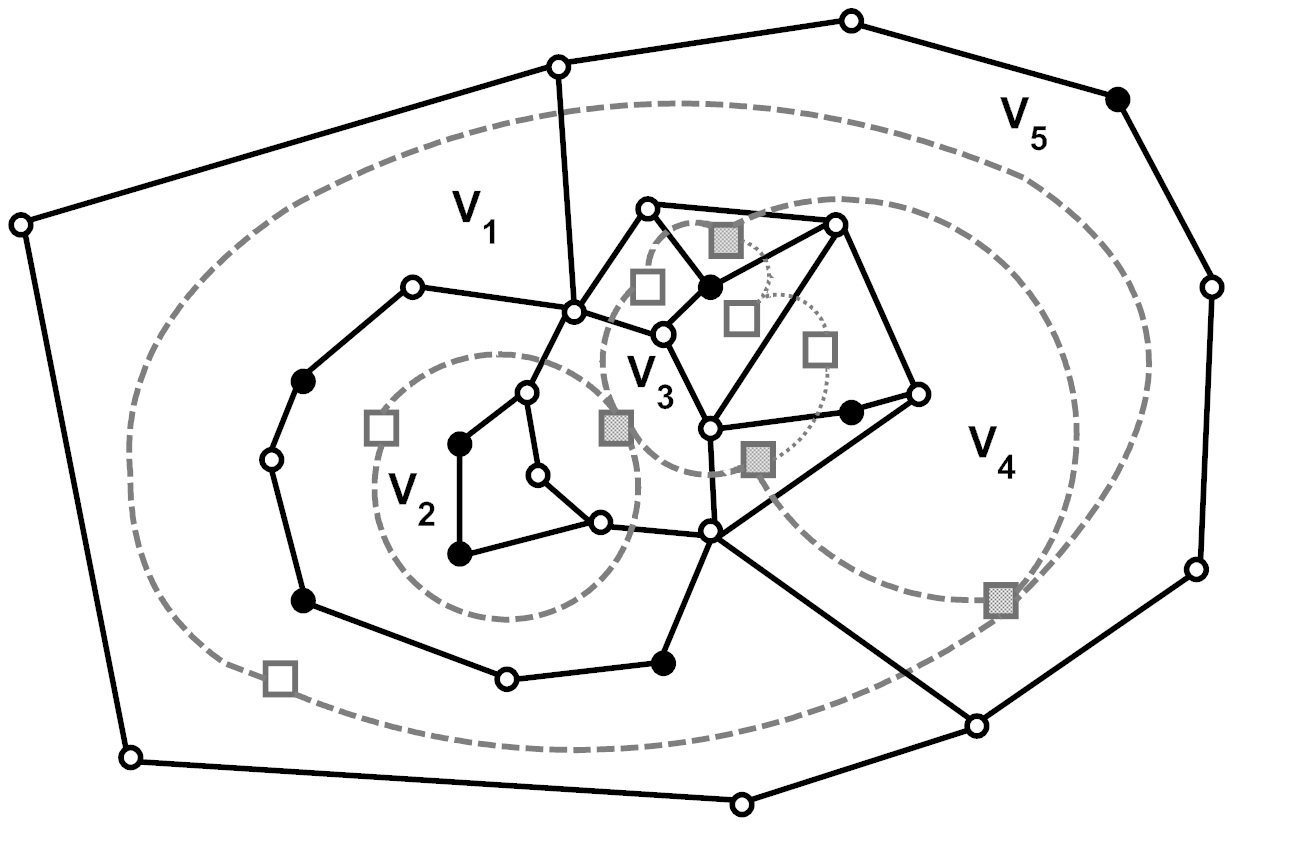}
    \end{center}
    \caption{A multi-cluster cut in a planar graph with five clusters. The edges of the initial (primal) graph are in plain lines, the non-terminal vertices are the white round vertices, the terminals are the black round vertices, the dual vertices are the square vertices, and the dual edges associated with the multi-cluster cut $\mathcal{C}$ are in dashed and dotted lines. (The edges of $\mathcal{C}_1$ are in dashed lines, and the four grey square vertices are the joint-vertices of $\mathcal{C}$.)}
    \label{figPlanarMinMCC}
\end{figure}

Given a \textsc{MinMCC} instance $I=(G=(V,E), w, \mathcal{T}=\{\mathcal{T}_1, \dots, \mathcal{T}_p\})$ and an optimal multi-cluster cut $C$ for $I$, we denote by $C^*$ the edge set dual to $C$, and, for each $i$, by $V_i$ the vertices of the connected component of $G'=(V, E \setminus C)$ containing the terminals in $\mathcal{T}_i$, and by $C_i$ the set of edges such that $C_i \subseteq C$ and $C_i$ has exactly one endpoint in $V_i$. We define a \emph{joint-vertex} as a dual vertex (a vertex of the dual graph $G^*$ of $G$) of degree at least 3 in $C^*$. Note that each $C_i$ corresponds to a set of (not necessarily simple) cycles in $G^*$. Let us assume for now that each $C_i$ corresponds to only one cycle.

If the edges in the embedding of the dual graph are viewed as curves in the plane (the dual vertices being intersections between curves), then the dual image of each $C_i$ will be a closed curve $\mathcal{C}_i$ (the union of all the $\mathcal{C}_i$'s, i.e., the geometric representation of $C^*$, will be denoted by $\mathcal{C}$); if this closed curve is simple (this may not be the case, see below), then, by the Jordan curve theorem, the faces of $G^*$ associated with all the terminals in $\mathcal{T}_i$ are \emph{inside} this curve, and the faces of $G^*$ associated with all the terminals in $\bigcup_{j\neq i} \mathcal{T}_j$ are \emph{outside} this curve (which simply means that the edges associated with $\mathcal{C}_i$ isolate the terminals in $\mathcal{T}_i$ from all the other terminals). When $\mathcal{C}_i \subset \mathbb{R}^2$ is not simple (as this is the case for $\mathcal{C}_1$ in Figure \ref{figPlanarMinMCC}), i.e., when $\mathcal{C}_i$ self-intersects in one or more points of the plane, the situation is a bit more complex: in this case, by a simple corollary of the Jordan curve theorem, $\mathbb{R}^2 \setminus \mathcal{C}_i$ contains more than two connected regions (a connected region of $\mathbb{R}^2 \setminus \mathcal{C}_i$ being a region of $\mathbb{R}^2$ such that any two points of this region can be linked by a curve without crossing $\mathcal{C}_i$), and one of these connected regions is unbounded (it is called the \emph{unbounded region}), while all the other ones are bounded. The only bounded region of $\mathbb{R}^2 \setminus \mathcal{C}_i$ (and all the faces it contains) that is adjacent to the unbounded region is called the \emph{inside} of $\mathcal{C}_i$ (it is unique since
$V_i$ is connected), and every other bounded region of $\mathbb{R}^2 \setminus \mathcal{C}_i$ is called an \emph{inner region} of $\mathcal{C}_i$ (although it does belong to the outside of $\mathcal{C}_i$, and not to its inside).

Notice that, if some $C^*_i$ contains more than one cycles (either simple or not), then either this means that there is one cycle $\bar{C}^*_i$ contained in $C^*_i$, corresponding to a closed curve $\bar{\mathcal{C}}_i$ in $G^*$, such that any other cycle contained in $C^*_i$ lies inside $\bar{\mathcal{C}}_i$, or this means that $V_i$ is the only component in contact with the infinite face. (In the first case, note that there is at least one other $C^*_j$ for some $j \neq i$ that lies inside the closed curve corresponding to each cycle in $C^*_i \setminus \bar{C}^*_i$.) So, we have:

\begin{lemma}\label{lemHomotopy1}
For each $i$, if $\mathcal{C}_i$ is a closed curve, then the faces associated with $\mathcal{T}_i$ are inside $\mathcal{C}_i$, while the faces associated with $\mathcal{T}_j$ are outside $\mathcal{C}_i$, for each $j \neq i$.
\end{lemma}

We also need to define \emph{homotopic curves}. Roughly speaking, given a set $\mathcal{O}$ of $\mu$ obstacles (typically, faces) $\mathcal{O}_1, \dots, \mathcal{O}_{\mu}$ in the plane, two simple curves $\mathcal{C}_1, \mathcal{C}_2$ in $\mathbb{R}^2 \setminus \mathcal{O}$ sharing the same endpoints (or two closed curves) are said to be \emph{homotopic with respect to $\mathcal{O}$} if $\mathcal{C}_1$ can be continuously deformed into $\mathcal{C}_2$ in $\mathbb{R}^2 \setminus \mathcal{O}$. We can also say that $\mathcal{C}_1$ is homotopic to $\mathcal{C}_2$ with respect to $\mathcal{O}$, or alternatively that $\mathcal{C}_1$ and $\mathcal{C}_2$ belong to the same homotopy class. In the present setting, the curves we will consider are the ones that are associated with (i.e., that are the dual images of) the $C_i$'s (or parts of them); the set of obstacles $\mathcal{O}$ we will consider is the set of faces associated with the terminals. Then, the following lemma is easy to see:

\begin{lemma}\label{lemHomotopy2}
Two closed curves having the same faces of $\mathcal{O}$ in their insides and the same faces of $\mathcal{O}$ in their outsides are homotopic with respect to $\mathcal{O}$.
\end{lemma}

Finally, let us notice that the number of vertices in $G^*$ is bounded by $2 \vert V \vert - 4$, since it is equal to the number of faces $f_G$ of $G$. Indeed, $G$ is a simple, loopless and connected planar graph, and hence each of its faces contains at least three vertices and edges: this implies that $2 \vert E \vert \geq 3 f_G$, which, combined with Euler's formula $\vert V \vert+f_G-|E| = 2$, yields $f_G \leq 2\vert V \vert - 4$. However, we still have to bound the number of joint-vertices in $C^*$. To this end, the following lemma will be useful in the next section:

\begin{lemma}\label{lemNbDualJoints}
The number of joint-vertices in $C^*$ is at most $2p-4$.
\end{lemma}
\begin{proof}
This can be shown by a simple application of Euler's formula. Consider the subgraph of $G^*$ induced by $C^*$. In this subgraph, there is no vertex of degree 1, and we contract any vertex of degree 2 in this subgraph (this does not modify the number of vertices of degree at least 3) in order to obtain the graph $G^*_C$. The number of faces in $G^*_C$ is $p$, since each cluster in $\{\mathcal{T}_1, \dots, \mathcal{T}_p\}$ induces exactly one connected component in $G$. We remove loops (and associated faces) as well as multiple edges (and associated faces) from $G^*_C$: each time we remove such an edge, we remove one face. If we denote by $m_C$ and $f_C$ the number of edges and faces in $G^*_C$, and by $n_C,m'_C,f'_C,\kappa_C$ the number of vertices, edges, faces, and connected components in this updated (simple) graph, respectively, then by Euler's formula we have $n_C+f'_C-m'_C=1+\kappa_C$, i.e., $n_C+f_C-m_C=1+\kappa_C$. (Note that $n_C$ is the number of joint-vertices we have to consider.) Any vertex in $G^*_C$ has degree at least 3, and hence $2 m_C \geq 3n_C$. Since $\kappa_C \geq 1$, we have $n_C+f_C-m_C\geq 2$, i.e., $n_C \geq m_C-f_C+2 \geq 3n_C/2-f_C+2$, and this yields $n_C/2 \leq f_C-2$, i.e., $n_C \leq 2f_C-4 = 2p-4$.\qed
\end{proof}

A similar result was presented in \cite[Theorem 5]{refYeh01}, using the notion of \emph{component graph} (in which there is a vertex for each component $V_i$ and a single edge between any two vertices if the corresponding components share at least one edge); however, a joint-vertex may actually \emph{not} induce a face in the component graph (see the joint-vertex belonging to $C^*_2$ in Figure \ref{figPlanarMinMCC} for instance), since this graph is simple by definition, and hence this proof was incomplete.

\section{Description and proof of the algorithm}\label{sectAlgo}

\subsection{A structural description of optimal solutions}\label{sectAlgo:Struct}

Dahlhaus et al. \cite{refDahlhaus94}, and later Hartvigsen \cite{refHartvigsen}, gave structural descriptions of optimal planar multiterminal cuts (one is based on the notion of \emph{topology} and on minimum spanning trees computation, and the other is based on links between optimal planar multiterminal cuts and Gomory-Hu cut collections). However, it is not clear whether these structural results could be extended to optimal planar multi-cluster cuts; in fact, it seems that they cannot. Here, we give a new and somewhat simpler structural description of optimal planar multiterminal cuts (although it may imply enumerating more elements than in the approaches described by Dahlhaus et al. and Hartvigsen), that is also valid for optimal planar multi-cluster cuts.

We use the definitions and notations from the previous section. Let $\mathcal{F}$ (resp. $\mathcal{F}_i$) be the faces of $G^*$ associated with the terminals in $\mathcal{T}$ (resp. in $\mathcal{T}_i$), and let $C$ be any multi-cluster cut that partitions the plane into $p$ connected regions (each one containing one cluster), such as a minimum multi-cluster cut (for instance). Let us now consider $C_i$ for some $i$, and assume that the dual image $\mathcal{C}_i$ of $C_i$ consists of only one closed curve. This curve goes through a certain number of joint-vertices: let us call them $\omega_1, \dots, \omega_{q_i}$, in clockwise order (with $\omega_1 = \omega_{q_i}$). Recall that, by definition, the curve $\mathcal{C}_i$ intersects other $\mathcal{C}_j$'s \emph{only} at joint-vertices. Assume that $q_i \geq 2$. Then, we have:

\begin{lemma}\label{lemStruct}
Let $V_i$ be a connected component of $G'=(V, E \setminus C)$, let $\mathcal{C}_i$ be the associated curve in $G^*$, and let $\omega_1, \dots, \omega_{q_i}$ be the joint-vertices $\mathcal{C}_i$ goes through. Then, $\mathcal{C}'=(\mathcal{C}\setminus \mathcal{C}_i) \cup \mathcal{C}'_i$ is also a valid multi-cluster cut for $I$, where $\mathcal{C}'_i$ is any cycle in $G^*$ going through $\omega_1, \dots, \omega_{q_i}$, and such that the faces associated with $\mathcal{T}_i$ are inside $\mathcal{C}'$, while the faces associated with $\mathcal{T}_j$ are outside $\mathcal{C}'$ for each $j \neq i$.
\end{lemma}
\begin{proof}
Assume that one such $\mathcal{C}'$ is not a multi-cluster cut. Consider any path $\mu_{a,b}$ in $G'=(V, E \setminus C')$ between two terminal vertices $t_a \in \mathcal{T}_j$ and $t_b \in \mathcal{T}_{j'}$ for some $j \neq j'$. We cannot have $j=i$ or $j'=i$, by the definition of $\mathcal{C}'_i$. Moreover, since $C$ is a multi-cluster cut, we know that $\mu_{a,b}$ contains at least one edge in $C_i$, say $uv$. Choose an edge dual to such an edge in $C_i$, and assume without loss of generality that this dual edge belongs to the curve $\mathcal{C}_i[\omega_1,\omega_2]$, defined as the part of $\mathcal{C}_i$ linking $\omega_1$ and $\omega_2$. From Lemma \ref{lemHomotopy2}, $\mathcal{C}'_i$ is homotopic to $\mathcal{C}_i$ with respect to $\mathcal{F}$. Hence, $\mathcal{C}_i$ can be continuously deformed into $\mathcal{C}'_i$ in $\mathbb{R}^2 \setminus \mathcal{F}$. In particular, since $\mathcal{C}'_i$ goes through $\omega_1$ and $\omega_2$, it contains some curve $\mathcal{C}'_i[\omega_1,\omega_2]$ homotopic to $\mathcal{C}_i[\omega_1,\omega_2]$ with respect to $\mathcal{F}$. Hence, the inside of the closed curve $\mathcal{C}_i[\omega_1,\omega_2] \cup \mathcal{C}'_i[\omega_1,\omega_2]$ contains neither $t_a$ nor $t_b$ (since $i,j,j'$ are all distinct). We claim the following :

\begin{Claim}
$\mu_{a,b}$ must ``intersect'' (i.e. have an edge in common with) $\mathcal{C}'_i[\omega_1,\omega_2]$ at least once.
\end{Claim}
\begin{proof}
Since $\mathcal{C}_i[\omega_1,\omega_2] \cup \mathcal{C}'_i[\omega_1,\omega_2]$ is a closed (but not necessarily simple) curve, the edge dual to any edge on its boundary either belongs to both $\mathcal{C}_i[\omega_1,\omega_2]$ and $\mathcal{C}'_i[\omega_1,\omega_2]$ (which is clearly not the case for $uv$, otherwise we are done), or has one endpoint inside $\mathcal{C}_i[\omega_1,\omega_2] \cup \mathcal{C}'_i[\omega_1,\omega_2]$ and one endpoint outside $\mathcal{C}_i[\omega_1,\omega_2] \cup \mathcal{C}'_i[\omega_1,\omega_2]$ (so, this is the case for $uv$).

Now, assume that $\mu_{a,b}$ has $t \geq 1$ (for some $t$) edges in common with $\mathcal{C}_i[\omega_1,\omega_2]$ (none of them is of the first type described above, otherwise we are done). If $\mu_{a,b}$ crosses $\mathcal{C}'_i[\omega_1,\omega_2]$, then we are done. Assume otherwise. $\omega_1$ and $\omega_2$ being two consecutive joint-vertices in $\mathcal{C}_i$, then by definition each of these $t$ edges has one endpoint in $V_i$ and the other one in $V_l$ for some $l$ (the same $l$ for all these edges). In particular, the vertices inside $\mathcal{C}_i[\omega_1,\omega_2] \cup \mathcal{C}'_i[\omega_1,\omega_2]$ that are incident to edges in $\mathcal{C}_i[\omega_1,\omega_2]$ all belong to the same connected component of $(V,E\setminus C)$ (either $V_i$ or $V_l$). Hence, each time $\mu_{a,b}$ ``crosses'' $\mathcal{C}_i[\omega_1,\omega_2]$, it ``changes side'' (going for instance from $V_i$ to $V_l$, then from $V_l$ to $V_i$, then again from $V_i$ to $V_l$, etc.). If it crosses $\mathcal{C}_i[\omega_1,\omega_2]$ an even number of times (the first edge crossed being $uv$ and the last one $u'v'$ for instance), then $u$ and $v'$ either both belong to $V_i$ or both belong to $V_l$ (i.e., belong to the same connected component of $(V,E\setminus C)$). So, instead, we can find a new path $\mu'_{a,b}$ from $t_a$ to $t_b$ that does not cross $\mathcal{C}_i[\omega_1,\omega_2]$ at all, by replacing the part of $\mu_{a,b}$ going from $u$ to $v'$ by a path from $u$ to $v'$ using vertices of $V_i$ (or $V_l$) only; this yields a contradiction. By the same argument, we can show that if $\mu_{a,b}$ crosses $\mathcal{C}_i[\omega_1,\omega_2]$ an odd number of times (the first edge crossed being $uv$ and the last one $u'v'$ for instance; note that $u'v'$ may be $uv$), then $v'$ is inside $\mathcal{C}_i[\omega_1,\omega_2] \cup \mathcal{C}'_i[\omega_1,\omega_2]$. Since the part of $\mu_{a,b}$ going from $v'$ to $t_b$ crosses neither $\mathcal{C}_i[\omega_1,\omega_2]$ (by definition) nor $\mathcal{C}'_i[\omega_1,\omega_2]$ (by assumption), and since  neither $t_a$ nor $t_b$ are inside $\mathcal{C}_i[\omega_1,\omega_2] \cup \mathcal{C}'_i[\omega_1,\omega_2]$, this yields a contradiction. Thus, $\mu_{a,b}$ must cross $\mathcal{C}'_i[\omega_1,\omega_2]$.\qed
\end{proof}

From this claim, $\mathcal{C}'$ intersects any path between two terminals lying in different clusters: it contradicts the fact that $\mathcal{C}'$ is not a multi-cluster cut.\qed
\end{proof}

We can then use this lemma to show that, if $q_i \geq 2$:

\begin{corollary}\label{corStruct}
Let $C$ be a minimum multi-cluster cut in a graph $G=(V,E)$, let $V_i$ be a connected component of $G'=(V, E \setminus C)$, let $\mathcal{C}_i$ be the associated curve in $G^*$, and let $\omega_1, \dots, \omega_{q_i}$ be the joint-vertices $\mathcal{C}_i$ goes through. Then, $\mathcal{C}_i$ is a shortest cycle in $G^*$, that is homotopic to any cycle $\Gamma$ in $G^*$ going through $\omega_1, \dots, \omega_{q_i}$ and being such that the faces in $\mathcal{F}_i$ are inside $\Gamma$, while the faces in $\mathcal{F}_j$ are outside $\Gamma$ for each $j \neq i$.
\end{corollary}
\begin{proof}
Assume that $\mathcal{C}_i$ is not such a shortest cycle. Then, we can replace $\mathcal{C}_i$ by a shortest cycle $\Gamma^*$ in $G^*$ going through $\omega_1, \dots, \omega_{q_i}$, and such that the faces in $\mathcal{F}_i$ are inside $\Gamma^*$, while the faces in $\mathcal{F}_j$ are outside $\Gamma^*$ for each $j \neq i$. From Lemma \ref{lemStruct}, $\mathcal{C}'=(\mathcal{C}\setminus \mathcal{C}_i) \cup \Gamma^*$ is also a valid multi-cluster cut for $I$. Moreover, $\Gamma^*$ is strictly shorter than $\mathcal{C}_i$ (since from Lemmas \ref{lemHomotopy1} and \ref{lemHomotopy2} they are homotopic with respect to $\mathcal{F}$), and hence $\mathcal{C}'$ is a strictly better solution than $\mathcal{C}$: a contradiction.\qed
\end{proof}

\subsection{Algorithmic aspects}

From Subsection \ref{sectAlgo:Struct}, we can construct $\mathcal{C}$ in an iterative way, by first ``guessing'' (i.e., enumerating) all the joint-vertices, then computing each $\mathcal{C}_i$ corresponding to a single closed curve one after the other, and finally removing the vertices inside it, and go on. (We assume without loss of generality that we look for an optimal solution having the maximum number of joint-vertices among the ones with $p$ clusters, and this implies that we cannot create ``new'' joint-vertices when computing each $\mathcal{C}_i$.) Hence, we have to guess an $i$ for which $\mathcal{C}_i$ corresponds to a single closed curve, compute $\mathcal{C}_i$ and remove it, and then go on by identifying another $i$ for which the part of $\mathcal{C}_i$ lying in the remaining graph (i.e., after removing the previous component) corresponds to a single closed curve, until there remains only one component. We can do this by enumerating all the possible sets of inclusions between the $C_i$'s (i.e., for each $i$ and $j\neq i$, whether there is one cycle $\bar{C}_i$ contained in $C_i$, that corresponds to a closed curve $\bar{\mathcal{C}}_i$ in the dual graph, and such that $C_j$ lies inside $\bar{\mathcal{C}}_i$; or whether there is one cycle $\bar{C}_j$ contained in $C_j$, that corresponds to a closed curve $\bar{\mathcal{C}}_j$ in the dual graph, and such that $C_i$ lies inside $\bar{\mathcal{C}}_j$; or finally whether none lies inside the other). Since the number of $C_i$'s is $p$ and since $p$ is fixed, this can be done in constant time.

In order to compute $\mathcal{C}_i$ for each $i$, we must first ``guess'' which joint-vertices $\mathcal{C}_i$ goes through (from Lemma \ref{lemNbDualJoints}, the maximum number of joint-vertices is $2p-4$, so guessing them requires to try all the possible ways of choosing at most $2p-4$ vertices among $2 \vert V \vert - 4$, which implies that the running time will depend on $p$), and then we can apply Corollary \ref{corStruct} and find a shortest cycle homotopic to some predefined curve in $G^*$ (keeping in mind that $\mathcal{C}_i$ may go through no joint-vertex; in this case, we only need to compute a minimum cut separating $\mathcal{T}_i$ from $\mathcal{T}_j$, for all $j \neq i$). (If needed, we can reduce the computation of a shortest cycle to the computation of a shortest path, by ``guessing'' the first edge of this path.) Finding a shortest homotopic path or cycle can be hard, if we require that it must be elementary; however, this property is not needed in our case. (And, indeed, some $\mathcal{C}_i$'s may be non simple cycles, such as $\mathcal{C}_1$ in Figure \ref{figPlanarMinMCC}.) We can compute a shortest homotopic path or cycle using for instance the algorithms given in \cite[Proposition 1]{Schrijver91} or in \cite{refECDV06}.

Finally, we have two last points to address. First, we must ensure that the shortest cycles or paths we compute go through predetermined joint-vertices. Second, we need to be able to generate all the possible predefined curves that the shortest paths we compute can be homotopic to. We now describe the strategy we use to deal with both points at the same time. Each time a given $\mathcal{C}_i$ goes through a given joint-vertex, this means that some vertices of the primal face associated with this joint-vertex belong to $V_i$. Actually, we even know that, on each face associated with a joint-vertex $\mathcal{C}_i$ goes through, there are at most $h_i+1$ sets of consecutive vertices (called \emph{intervals}) that belong to $V_i$, where $h_i \leq p$ is the number of inner regions of $\mathcal{C}_i$. Therefore, to the joint-vertices associated with a given $\mathcal{C}_i$ corresponds a set $B_i$ of distinct vertices of $V_i$ lying on the primal faces associated with these joint-vertices. The best way to encode this set $B_i$ is to include two vertices of each interval. For a given interval lying on the face associated with a given joint-vertex, call $a$ and $b$ the two vertices of this interval. Then, the vertices in $B_i$ associated with that interval are all the vertices of this face encountered while traveling clockwise from $a$ to $b$ on this face. Let us denote by $\mathcal{B}_i$ the set of dual faces associated with the vertices in $B_i$. By definition, each face associated with a joint-vertex contains at least two vertices belonging to two different $B_i$'s, thus from Lemma \ref{lemHomotopy2}, for each $i$, $\mathcal{C}_i$ is homotopic, with respect to the faces in $\mathcal{F}$ and $\bigcup_j \mathcal{B}_j$, to any closed curve being such that the faces in $\mathcal{F}_i \cup \mathcal{B}_i$ are inside it, and the faces in $\mathcal{F}_j \cup \mathcal{B}_j$, for each $j \neq i$, are outside it. More generally, any closed curve goes through the same joint-vertices as $\mathcal{C}_i$, if this curve is such that the faces in $\mathcal{B}_i$ belong to its inside, and the faces in $\mathcal{B}_j$ belong to its outside, for each $j \neq i$.

Since for each $i$ the inside of $\mathcal{C}_i$ is a connected region, i.e., the subgraph of $G$ induced by $V_i$ is connected, we also know that in $G'=(V, E \setminus C)$ all the vertices in $B_i$, as well as all the terminals in $\mathcal{T}_i$, are connected together. This implies that, for each $i$, we can construct a closed curve $\mathcal{C}'_i$ homotopic to $\mathcal{C}_i$ by choosing some tree spanning both $\mathcal{T}_i$ and $B_i$, and then removing the edges having exactly one endpoint in the $i$th of these spanning trees. (For each $i$, $\mathcal{C}'_i$ goes through the same joint-vertices as $\mathcal{C}_i$, and $\mathcal{C}'_i$ and $\mathcal{C}_i$ are indeed homotopic with respect to the faces in $\mathcal{F}$, since $\mathcal{T}_i$ is the only cluster that belongs to the inside of $\mathcal{C}'_i$, i.e., $C'_i$ isolates $\mathcal{T}_i$ from $\mathcal{T}_j$, for all $j \neq i$.) In practice, we have to ``guess'' $B_i$ for each $i$ (which, as mentioned above, can be done by enumerating at most two vertices of $G$ for each interval), making sure that the $B_i$'s define a partition of the vertex set of the faces associated with all the joint-vertices, and then construct $p$ vertex-disjoint trees (each one spanning $\mathcal{T}_i$ and $B_i$ for some $i$), and finally remove the edges isolating each tree from the rest of the graph. For each combination of $B_i$'s, finding such vertex-disjoint trees can be done in polynomial time (since the graph is planar, $\sum_{i=1}^p \vert \mathcal{T}_i \vert$ is fixed, and the number of mandatory vertices that the $p$ trees must span lie on at most $\sum_{i=1}^p \vert \mathcal{T}_i \vert+(2p-4)$ faces), using for instance the algorithm given in \cite[Theorem 4]{Schrijver91}.\\

So, our algorithm for planar \textsc{MinMCC} is as follows:

\begin{enumerate}
\item For each possible clustering of the terminals, for each possible set of inclusions between the $C_i$'s, for each possible combination of joint-vertices, and for each possible choice of the $B_i$'s do:
\begin{enumerate}
\item Compute $p$ vertex-disjoint trees, each spanning $\mathcal{T}_i$ and $B_i$ for some $i$, and construct the curves $\mathcal{C}'_i$ by removing, for each $i$, each edge incident to exactly one vertex of the $i$th tree;
\item For each $i$ except the last one (in the order given by the current set of inclusions, starting from a $\mathcal{C}_i$ including no other $\mathcal{C}_j$ for $j \neq i$), compute a shortest cycle homotopic to $\mathcal{C}'_i$ with respect to $\mathcal{F}$ and $\bigcup_j \mathcal{B}_j$; then, remove the vertices of the connected component of $G$ lying inside this cycle.
\end{enumerate}
\item Output the best feasible solution found.
\end{enumerate}

We already explained why all steps run in polynomial time, and it should be clear from our above discussion that this algorithm is correct. This yields:

\begin{theorem}
\textsc{MinMCC} can be solved in polynomial time in planar graphs, if the sum of the sizes of the clusters is fixed.
\end{theorem}

Therefore, we can finally state:

\begin{corollary}
\textsc{MinMC} can be solved in polynomial time in planar graphs, if the number of source-sink pairs is fixed.
\end{corollary}

\section{Acknowledgements}

The author would like to thank Sylvie Poirier for her help, and \'Eric Colin de Verdière for fruitful discussions on shortest homotopic paths.

\end{document}